\documentclass[aps,pra,10pt,superscriptaddress,twocolumn,nofootinbib]{revtex4-2}
\usepackage{graphicx}
\usepackage{amsmath,amsthm,amssymb,amsfonts,dsfont,mathtools,csquotes,xcolor,soul}
\usepackage{mdframed}
\usepackage{orcidlink}

\newtheorem{definition}{Definition}

\newtheorem{assumption}{Assumption}

\newtheorem{thm}{Theorem}
\newtheorem{Lemma}{Lemma}

\usepackage{enumitem}

\newcommand{\ket}[1]{\left| #1 \right\rangle}

\usepackage{hyperref}
\newcommand{\etal}{\textit{et al} }

\begin{document}

\title{Limits of Absoluteness of Observed Events in Timelike Scenarios: A No-Go Theorem}

\author{Sumit Mukherjee}
\email{mukherjeesumit93@gmail.com}
\affiliation{Department of Physics, Indian Institute of Technology Dharwad, Dharwad, Karnataka 580011, India}
\affiliation{Department of Physics of Complex Systems, S. N. Bose National Center for Basic Sciences, Block JD, Sector III, Salt Lake, Kolkata 700106, India}

\author{Jonte R. Hance\,\orcidlink{0000-0001-8587-7618}}
\email{jonte.hance@newcastle.ac.uk}
\affiliation{School of Computing, Newcastle University, 1 Science Square, Newcastle upon Tyne, NE4 5TG, UK}
\affiliation{Quantum Engineering Technology Laboratories, Department of Electrical and Electronic Engineering, University of Bristol, Woodland Road, Bristol, BS8 1US, UK}

\begin{abstract}
Wigner’s Friend–type paradoxes challenge the assumption that events are absolute — that when we measure a system, we obtain a single result, which is not relative to anything or anyone else. These paradoxes highlight the tension between quantum theory and our intuitions about reality being observer-independent. Building on a recent result that developed these paradoxes into a no-go theorem, namely the \emph{Local Friendliness Theorem}, we introduce the \emph{Causal Time-Symmetric Friendliness Paradox}, a time-ordered analogue of it. In this framework, we replace the usual locality assumption with \emph{Axiological Time Symmetry} (ATS), and show that, when combined with the assumptions of \emph{Absoluteness of Observed Events} (AOE), \emph{No Retrocausality} (NRC), and \emph{Screening via Pseudo Events} (SPE), we obtain a Causal Time-Symmetric Friendliness inequality. We then show that quantum mechanics violates this inequality and is therefore incompatible with at least one of these assumptions. To probe which assumption might be incompatible, we then examine whether AOE in its entirety is essential for this no-go result. We propose a weaker, operational form of AOE that still leads to inequalities which quantum mechanics violates. This result shows that even under relaxed assumptions, quantum theory resists reconciliation with classical notions of absolute events, reinforcing the foundational significance of Wigner’s Friend-type paradoxes in timelike scenarios.
\end{abstract}

\maketitle

\section{Introduction}

One of the most unsettling quantum paradoxes proposed so far is that of Wigner's Friend, which asks the question of what happens when, as should be valid in standard quantum mechanics, we treat human observers as quantum systems \cite{wigner1967remarks}. While originally simply troubling to our intuitions, in the same way as Schr\"odinger's Cat \cite{Schrodinger1935Cat} or the EPR paradox \cite{Einstein1935EPR}, recent work has shown that, by combining an extension of the Wigner's Friend scenario with certain intuitive assumptions (that any event an observer sees is real and absolute, and so alternative, or ``parallel'', versions of such observations do not occur simultaneously; that our choices are free and not influenced by other factors; and that no information can travel faster than $c$, the speed of light in a vacuum), we obtain predictions for the values of certain observables in the scenario which differ from the values predicted by standard quantum mechanics \cite{Bong2020ExtendedWF}. While these scenarios are practically impossible to perform in the real world, given they involve being able to unitarily-undo the process of a human observer measuring a quantum system, it is deeply interesting that we are able to obtain such a contradiction, especially given the assumptions are arguably weaker than those used in similar no-go results, such as Bell's Theorem \cite{Bell1964,Clauser1969CHSH,bell1971introduction,Freedman1972,Clauser1974CH74,Wiseman2017Causarum,Wharton2020Retrocausality,Hance2022ComNatPhys,Hance2024CFRestrictions} or the Kochen-Specker Theorem \cite{Kochen1968,Budroni2022ContextualityReview}.

However, it is still an open question whether any of the assumptions behind the so-called Local Friendliness inequalities~\cite{Bong2020ExtendedWF} in this Extended Wigner's Friend scenario can be weakened any further. In this paper, we show that the assumptions can be weakened, in two ways.

Firstly, we show that we can replace the Locality assumption (the assumption that no information travels faster than $c$) in this standard Extended Wigner's Friend scenario with an assumption of time-symmetry in a time-ordered version of the scenario - i.e., by making the scenario causally rather than spatially separated, we can weaken the assumptions by just assuming the choices in the first part of the experiment always affect the probabilities of the second part of the experiment in the same way, regardless of which part of the experiment is done first and which is done second.

Secondly, we show that in this time-ordered scenario, we can also weaken the assumption of Absoluteness of Observed Events (that any event an observer sees is real and absolute, meaning no ``parallel'' versions of such an observation exist with different measurement results), by an operational analogy of the weakening of the ontological assumption in Price's no-go result \cite{Price2012Retro} by Leifer and Pusey \cite{leifer2017}.

One counterargument against the intuitiveness and viability of the assumptions used to generate the Local Friendliness inequalities is that Absoluteness of Observed Events, and Unitary Undoing of an Observer's Measurement, are fundamentally incompatible, regardless of any other assumptions. Showing the inequality still holds for the weakened form of AOE implies that this counterargument fails to account for the entirety of the peculiarity of quantum mechanics' ability to violate these Local (or here, Causal) Friendliness inequalities.

This paper is laid out as follows. In Section~\ref{Sect:WF}, we first go through the original (Section~\ref{Subsect:OWF}) and extended (Section~\ref{Subsect:EWF}) Wigner's Friend thought experiments, with special focus for the Extended Wigner's Friend thought experiment on the assumptions required to achieve a contradiction with standard quantum mechanics.

In Section~\ref{Sect:CFP}, we extend these results to a new paradox: the Causal Time-Symmetric Friendliness Paradox. In Section~\ref{Subsect:NewScenario}, we describe the (hypothetical) experimental scenario for the Causal Time-Symmetric Friendliness Paradox to occur, before in Section~\ref{Subsect:CFNogo} giving the assumptions required for such a scenario to give predictions contradicting standard quantum mechanics (a contradiction we formally prove in the Appendix).

In Section~\ref{Sect:CFPWeakenedAOE}, we expand on our result from Section~\ref{Sect:CFP}, by seeing to what extent we can weaken the Absoluteness of Observed Events assumption common to both our Causal Time-Symmetric Friendliness Paradox and the Extended Wigner's Friend Paradoxes. We do this by looking at a result by Leifer and Pusey \cite{leifer2017}, where they weakened the $\psi$-ontic reproducibility assumption in Price's 2007 no-go result on time-symmetry and retrocausality \cite{Price2012Retro}, to instead just require $\lambda$-mediation. In Section~\ref{Subsect:lambdafrompsi}, we review how Leifer and Pusey weakened the $\psi$-ontic reproducibility assumption into $\lambda$-mediation. In Section~\ref{Subsect:WeakeningAOE}, we then analogise this weakening at an operational, rather than ontological, level, to enable us to weaken the Absoluteness of Observed Events assumption while still being able to generate an inequality that quantum mechanics violates. However, in Section~\ref{Subsect:impossAPE}, we then show that we cannot weaken this assumption any further (i.e., specifically remove the Absoluteness of the Pseudo Events - the Friends measuring their system - which may or may not be unitarily ``rewound'' by the Wigners) without losing the ability to generate such a contradiction, with either quantum correlations or even Popescu and Rohrlich's Boxworld's correlations \cite{Popescu1994}.

Finally, in Section~\ref{Sect:Discussions}, we discuss our results: in Section~\ref{Subsect:Comparisons}, we compare our results to previous work; and in Section~\ref{Subsect:Implications}, we discuss the implications of our work for the study of Wigner's Friend-style paradoxes, and quantum foundations more broadly.

\section{Wigner's Friend Scenarios and their implications}\label{Sect:WF}

\subsection{The Wigner's Friend paradox}\label{Subsect:OWF}

One of the peculiarities implied by the measurement problem \cite{Maudlin1995,Sclosshauer2005measurementProbem, Hance2022MeasProb} was illustrated by Wigner in his eponymous thought experiment, the Wigner's Friend paradox  \cite{wigner1967remarks}. The paradox involves a scenario with two observers, Wigner and his Friend, each capable of performing specific operations on the quantum systems assigned to them. While the Friend remains inside a laboratory $\mathbf{L}$, isolated from the external environment and measuring a quantum system $\mathbf{S}$ contained within, Wigner stays outside and treats the entire lab as a single quantum system.

The quantum system $\mathbf{S}$ inside the lab is described by a wave function
\begin{equation}
\ket{\psi}_{S}=\frac{\ket{\uparrow}_{S}+\ket{\downarrow}_{S}}{\sqrt{2}},
\end{equation}
where $\{\ket{\uparrow}_{S},\ket{\downarrow}_{S}\}$ are eigenstates of spin of a system in some arbitrary direction. Therefore, if the Friend measures the system, she gets an outcome corresponding to either $\ket{\uparrow}_{S}$ \textit{or} $\ket{\downarrow}_{S}$. However, from Wigner's perspective, the entire laboratory behaves as a large quantum system. This includes not only the system being measured by the Friend, but also the measurement device, whose corresponding quantum state can be either $\ket{\uparrow}_{D}$ or $\ket{\downarrow}_{D}$ depending on the state of the system, as well as the Friend, whose quantum state can similarly be $\ket{\uparrow}_{F}$ or $\ket{\uparrow}_{F}$. Altogether, there exists a linear map $f$ from the wave function of the measured system to the combined wave function of the full laboratory, given by:
\begin{equation}
  f:\left \{
    \begin{aligned}
    &\ket{\uparrow}_{S} \rightarrow \ket{\uparrow}_{S} \otimes \ket{\uparrow}_{D} \otimes \ket{\uparrow}_{F} \equiv \ket{\uparrow}_{L} \\
    & \ket{\downarrow}_{S} \rightarrow \ket{\downarrow}_{S} \otimes \ket{\downarrow}_{D} \otimes \ket{\downarrow}_{F} \equiv \ket{\downarrow}_{L}
\end{aligned}
\right.  
\end{equation}

Therefore, if the initial state of the system is $\ket{\psi}_{S}$, then after the Friend's measurement, the wave function of the lab (from Wigner's perspective) takes the form,
\begin{equation}
\label{wigner_wavefuction}
    \ket{\psi}_{W}= \frac{\ket{\uparrow}_{L}+\ket{\downarrow}_{L}}{\sqrt{2}}
\end{equation}

As the lab is isolated from the environment, Wigner has no way to know the result of the measurement that the Friend has already performed. Therefore, to Wigner, the state of the system remains as a probabilistic mixture of $\ket{\uparrow}_{S}$ and $\ket{\downarrow}_{S}$, which he gets by tracing out the device and Friend's wave function. In contrast, the Friend has already obtained a definite outcome (so to them, the state of the system is either in the pure state $\ket{\uparrow}_{S}$ or the pure state $\ket{\downarrow}_{S}$). This apparent contradiction between Wigner’s and his Friend’s perceptions (and descriptions) of a single event highlights the incommensurability of quantum theory in accounting for observer-independent facts.

However, on opening the Friend's lab, Wigner (and the wider environment) also becomes entangled with the lab, and so the Friend's system. This means, on opening the lab, the state of the system (from Wigner's perspective) also becomes a single eigenstate (either the pure state $\ket{\uparrow}_{S}$, or the pure state $\ket{\downarrow}_{S}$). Therefore, despite the difference in the state that Wigner and the Friend each associate to the Friend's measured system while the lab is closed, there is no point at which these two different states give predictions which could observably contradict. One could simply dismiss any apparent contradiction by asserting that \textit{unperformed experiments have no results} \cite{peresUnperformed}. This led to the question of what assumptions would be necessary for us to adapt the scenario in such a way that experimentally-testable predictions derived from these assumptions contradicted the predictions of quantum theory.

\subsection{Extended Wigner Friend Scenarios}\label{Subsect:EWF}

\begin{center}
    \begin{figure}
     \includegraphics[height=50mm, width=85mm,scale=4]{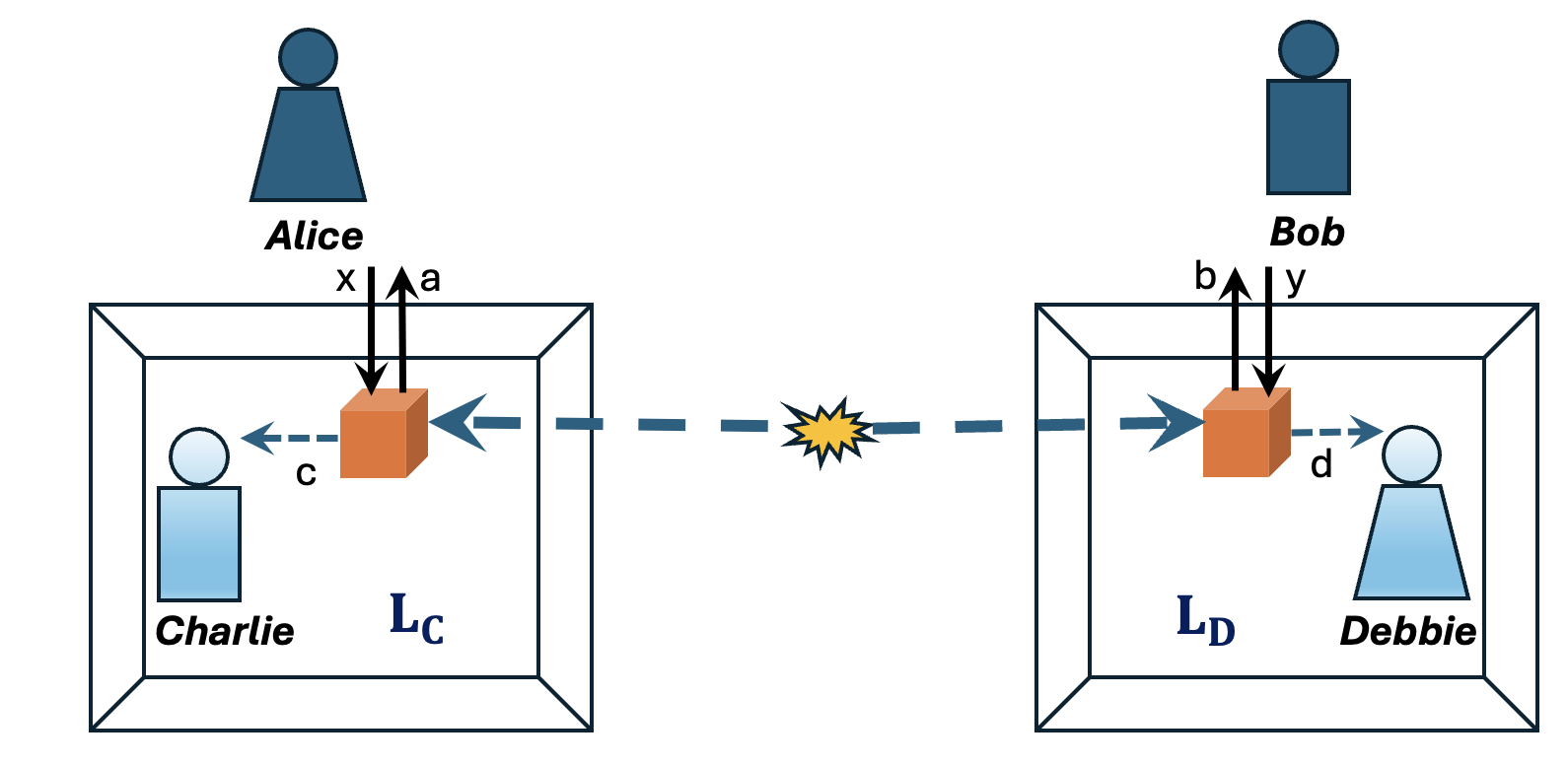}
    \caption{The extended Wigner's Friend Scenario for the Local Friendliness no-go theorem.}
    \label{fig:LFEWF}
\end{figure}
\end{center}

Recently, a flurry of works has emerged \cite{brukner2015quantummeasurementproblem,Brukner2018Observer,Healey2018-ur,Frauchiger2018Consistently,Lazarovici2019-ol,sarkar2023,Allard_Guerin2021-yc,marek21,santo25}, building upon the original Wigner’s Friend paradox \cite{wigner1967remarks} and its subsequent generalisations \cite{Deutsch1985}, and exploring various foundational issues related to these scenarios. Frauchiger and Renner \cite{Frauchiger2018Consistently} (potentially inspired by Brukner \cite{brukner2015quantummeasurementproblem}) first proposed an extension of the Wigner's Friend paradox in which the quantum probability assignments are shown to be logically inconsistent with the assumption of observer-independent facts. However, it was only in Brukner's later version \cite{Brukner2018Observer} that such inconsistencies were used to derive differing, testable predictions. Although these works received significant attention, their conclusions rely on counterfactual reasoning---an aspect which arguably weakened their claims. Bong \textit{et al} later put forward a stronger no-go theorem \cite{Bong2020ExtendedWF}, which, rather than challenging the notion of observer-independent facts involving counterfactual scenarios, directly questions the absoluteness of already observed events. 
 
Inspired by Brukner's proposal \cite{Brukner2018Observer}, Bong \etal ~\cite{Bong2020ExtendedWF} extended the Wigner's Friend thought experiment to the scenario where two observers (Charlie and Debbie) each possess one part of a non-trivially correlated pair (i.e., in quantum theory, an entangled state). These two observers are in labs $\mathbf{L}_{C}$ and $\mathbf{L}_{D}$ respectively. Further, there exist two more observers, Alice and Bob, who can perform any arbitrary operation on lab $\mathbf{L}_{C}$ and $\mathbf{L}_{D}$ respectively. As they each have this power of performing measurements over one of the labs (and so one of the observers), they are referred to as super-observers. This is pictorially depicted in Fig. \ref{fig:LFEWF}.

In this scenario, Charlie measures his system (staying inside his lab), obtaining output $c \in \{0,1\}$. This measurement process is describable by a unitary operation $\mathcal{U}_{C}$ on the lab $\mathbf{L}_{C}$. After Charlie's measurement, Alice chooses $x \in \{0,1,2,..\}$ (her input) at random, acting depending on the value of $x$. If $x=0$, then Alice opens lab $\mathbf{L}_{C}$ and asks Charlie about his measurement result, setting her output $ a\in \{0,1\}$ equal to Charlie's output $c$. In other cases (i.e., if $x=1,2,..$), Alice applies $\mathcal{U}_{C}^{\dagger}$ to $\mathbf{L}_C$, undoing Charlie's measurement of the system, before opening lab $\mathbf{L}_{C}$, and measuring Charlie's system herself, setting her output $a$ as the result of this measurement. Debbie and Bob proceed in the same manner as Charlie and Alice on the other side of the experimental set-up, except Bob uses random input $y\in \{0,1,2,...\}$ instead of $x$, and Debbie's and Bob's outcomes are denoted by $d\in \{0,1\}$ and $b\in \{0,1\}$, respectively (rather than $c$ and $a$). 

Bong \etal derived what they call Local Friendliness (LF) inequalities by considering the empirical probabilities $p(a,b|x,y)$ we get for this scenario when we impose the following set of assumptions:

\begin{itemize}
    \item \textit{Assumption 1. Absoluteness of Observed Events (AOE): An observed event is a real single event, and not relative to anything or anyone.}

This implies that there exists a joint probability distribution $p(a,b,c,d|x,y)$ such that all the empirical probabilities can be recovered from it, as

    \begin{equation}
        p(a,b|x,y)=\sum_{c,d}p(a,b,c,d|xy) \ \ \forall a,b,x,y
    \end{equation}

    \item \textit{Assumption 2. No-Superdeterminism (NSD): Any set of unitarily-erasable events on a space-like hypersurface is uncorrelated with any set of freely chosen actions subsequent to that space-like hypersurface.}

Here, this implies that the joint probability of outcomes $c$ and $d$ is independent of $x$ and $y$, i.e., $p(c,d|xy)=p(c,d)$. In the Bell scenario, this assumption is known as Statistical Independence \cite{Hance2024CFRestrictions} (often referred to erroneously as the ``Free Will'' assumption), which says that the hidden variables are conditionally independent of the choice of measurements. However, in contrast to the Bell scenario, here the hidden variable-like values $c$ and $d$ are outcomes of (unitarily-erasable) observed events.
    
    \item \textit{Assumption 3. Locality (L): The probability of an observable event $e$ is unchanged by conditioning on a space-like-separated free choice $z$, even if it is already conditioned on other events (so long as these events are not in the future light-cone of $z$).}

For the current set-up, the locality assumption reads:

    \begin{align}
        p(a|c,d,x,y)=p(a|c,d,x) \ \ \forall c,d,x,y \nonumber \\
        p(b|c,d,x,y)=p(b|c,d,y) \ \ \forall c,d,x,y 
    \end{align}
\end{itemize}
This assumption is similar to the Parameter Independence assumptions used in the conventional Bell scenario \cite{Myrvold2024sep-bell-theorem,Shimony_1993} - the only difference being that here the (potentially-observed) events $c$ and $d$ replace the hidden variable $\lambda$ from the Bell scenario.

Bong \etal showed that quantum theory violates the derived inequality, and from this violation they claimed: as Assumptions 2 and 3 are (supposedly) satisfied by quantum statistics, then Assumption 1 must be violated. They therefore argue that quantum theory is incompatible with the assumption of Absoluteness of Observed Events. 

It is important to note that, even though the events $c$ and $d$ are treated as observed events in Bong \etal's treatment, this is not necessarily the case: when $x\neq0$ ($y\neq0$), event $c$, $d$ is unitarily erased by Alice (Bob), as if a value was never assigned to the variable. Thus, we argue that events $a$,$b$ and events $c$, $d$ cannot be treated in the same way. We therefore call $c$, $d$ \emph{Pseudo Events}, while $a$ and $b$ we refer to as \emph{Truly-Observed Events}.

We emphasise that the terms Pseudo Events and Truly-Observed Events should not be interpreted as universally- or ontologically-defined concepts. Rather, they are scenario-dependent notions introduced solely to specify the operational structure of the Causal Time-Symmetric Friendliness setup. More precisely, an event is labelled a Pseudo Event if, within the prescribed experimental protocol, there is ever the possibility of it being unitarily erased or being made operationally inaccessible by a super observer, while an event is labelled Truly-Observed event if it is never erased, and always remains empirically accessible in the scenario under consideration. This distinction is therefore relative to the chosen protocol and the level of control available to the (super)observers, and should not be interpreted as an absolute or fundamental classification of events. While it could be argued that, in quantum mechanics, all events are unitarily erasable, so all events are Pseudo Events, this has not been empirically demonstrated. Therefore, in the absence of any way to empirically demonstrate all events are unitarily erasable, we treat Truly-Observed Events as the rule and Pseudo Events as the exception, rather than the reverse.

\begin{center}
    \begin{figure}
     \includegraphics[height=50mm, width=85mm,scale=4]{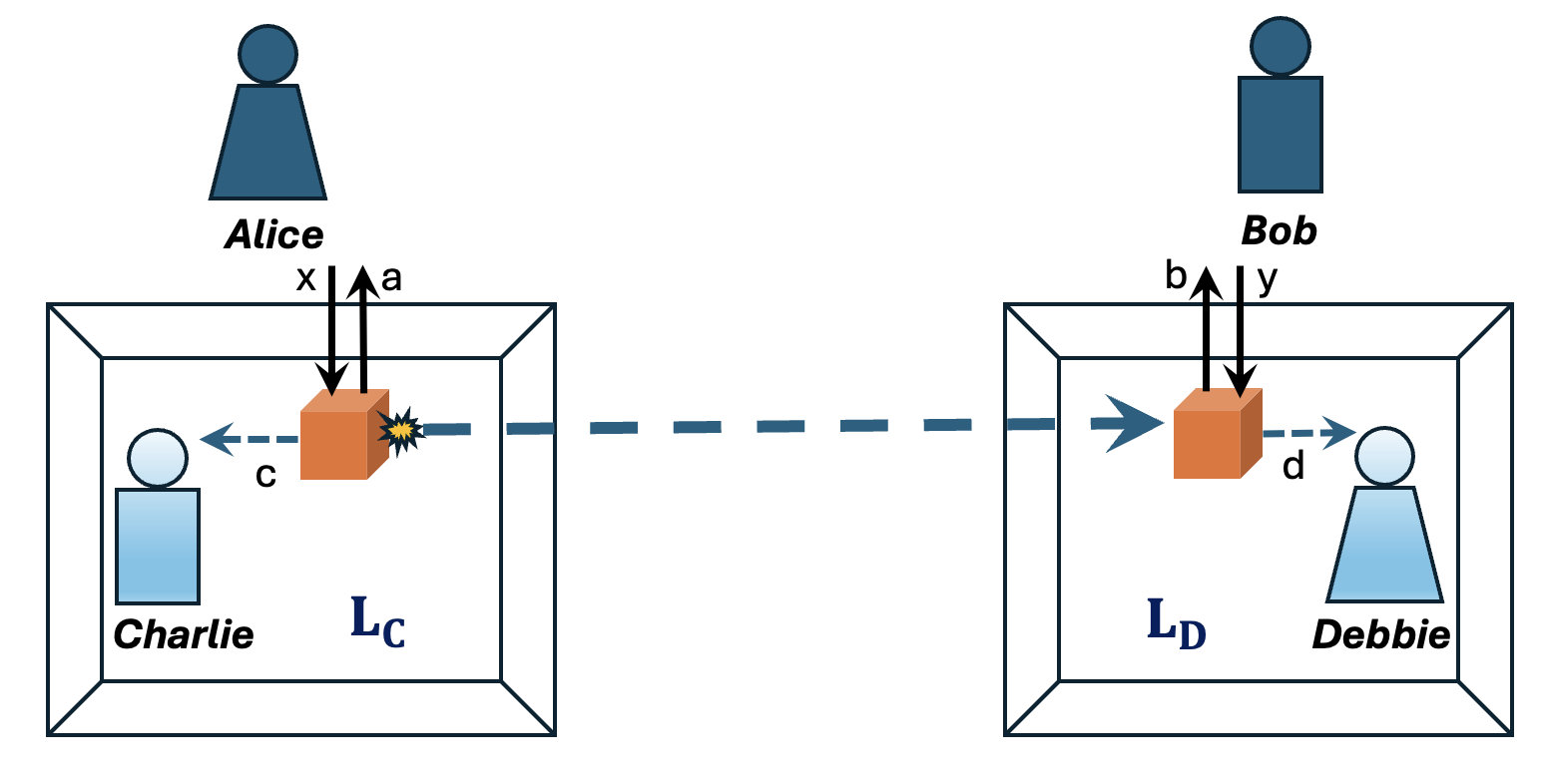}
    \caption{Time-ordered measurement scenario for Causal Time-Symmetric Friendliness no-go theorem.} 
    \label{fig:CFEWF}
\end{figure}
\end{center}

\section{Causal Time-Symmetric Friendliness Paradox}\label{Sect:CFP}
\subsection{The Causal Time-Symmetric Friendliness Scenario}\label{Subsect:NewScenario}

Here, we introduce our new scenario, illustrated in Fig.~\ref{fig:CFEWF}. As in Ref.~\cite{Bong2020ExtendedWF}, Charlie and Debbie begin the experiment sealed in their respective labs. However, here they do not share an entangled state. Instead, Charlie measures his system first, with his measurement outcome denoted by $c \in \{0,1\}$. Given his lab is completely closed to the environment, we again assume this measurement process can be represented by the unitary action $\mathcal{U}_{C}$ from the point of view of the Alice, who resides outside the lab and acts as an super-observer. After Charlie's measurement, Alice acts based on her random input $x \in \{0,1\}$: if $x=0$, Alice opens Charlie's lab, and sets her outcome $a$ equal to his outcome $c$; if $x=1$, Alice unitarily reverses Charlie's measurement process (i.e., applies unitary $\mathcal{U}^{\dagger}_{C}$ to lab $\mathbf{L}_C$) before opening Charlie's lab, and performs some measurement of her own on the system, setting her outcome $a$ as the result of her measurement. 

After Alice sets her outcome $a$, the system Charlie (if $x=0$) or Alice (if $x\neq0$) just interacted with is sent into Debbie's lab, where Debbie measures the system and gets the measurement result $d$. This measurement is characterised by the unitary $\mathcal{U}_{D}$ from outside of the lab. After Debbie's measurement, the super-observer Bob acts based on his random input $y$: if $y=0$, Bob opens Debbie's lab, setting his measurement outcome $b$ to be equal to Debbie's outcome $d \in \{0,1\}$; whereas if $y=1$, Bob applies $\mathcal{U}^{\dagger}_{d}$ to unitarily reverse Debbie's measurement of the system, before opening Debbie's lab and performing a measurement of his own on the system, setting his outcome $b$ equal to the measurement result.

We define a \emph{behaviour} as the collection of joint probability distributions associated with all events in a given experimental scenario, including both Truly-Observed Events and Pseudo Events. A behaviour is specified by the set of joint probabilities $ \mathcal{P} := \{p(c,a,d,b | x,y)\}_{a,b,c,d,x,y}$ for all choices of inputs $x,y$, from which all operationally accessible marginals and conditional probabilities can be obtained by marginalisation. Further, we will call a behaviour $ \mathcal{S} $ that satisfies certain specific properties a \textit{sector} if $\mathcal{S} \subseteq \mathcal{P} $.

\subsection{The Causal Time-Symmetric Friendliness No-Go Theorem}\label{Subsect:CFNogo}

The main metaphysical assumption whose validity we aim to examine here is the same one used to establish the Local Friendliness inequality in \cite{Bong2020ExtendedWF}: the assumption of \textbf{Absoluteness of Observed Events}. However, the no-go theorem we intend to establish is formulated within the prepare-measure scenario discussed in the previous section and illustrated in Fig.~\ref{fig:CFEWF}. In this setting, the aforementioned assumption can be phrased as follows:

\begin{assumption}{Absoluteness of Observed Events:}
  All observed events at a particular point on the space-time region characterised by the measurement outcomes have definite values, irrespective of any event that occurred at any other point in the space-time. 
    
    In Causal Time-Symmetric Friendliness scenario this implies for measurement events  $a,b,c,d$ the joint probability $p(c,a,d,b|x,y)$ exists and the observed probability $p(a,b|x,y)$ can be recovered from it via marginalisation as,
    \begin{equation}
        p(a,b|x,y) = \sum_{c,d} p(c,a,d,b|x,y) \ \ \forall x,y,a,b.
    \end{equation} 
\end{assumption}

As discussed, the scenario we consider here can be cast as a conventional prepare-measure scenario consisting of a preparation stage $\mathbb{P}$ and a measurement stage $\mathbb{M}$. Each stage involves a specified set of inputs, which are directly controlled by the experimentalist, together with a corresponding set of outputs for each input, which are not directly controlled by the experimentalist. Building upon our description of observed events, and inspired by Refs.~\cite{leifer2017,Price2012Retro}, we first define:

\begin{definition}[Operational Time Symmetry] A theory satisfies operational time symmetry if for every prepare-measure experiment with preparation stage $\mathbb{P}$ and measurement stage $\mathbb{M}$, there exists a time-reversed experiment, with preparation stage $\mathbb{P}^{'}$ and measurement stage $\mathbb{M}^{'}$ (with same number of inputs and outputs as that of $\mathbb{M}$ and $\mathbb{P}$ respectively), such that joint probabilities of observed events remain the same, i.e.,
\begin{equation}
    p_{\leftarrow}(f,e|v,u)=  p_{\rightarrow}(e,f|u,v) \ \forall u,v,e,f.
\end{equation}
where $e$ and $f$ are outcomes corresponding to observed events when the measurements $u$ and $v$ respectively are performed, while $p_{\rightarrow}(....)$ and $p_{\leftarrow}(....)$ are probabilities for forward and time-reversed event respectively.
\end{definition}

Less formally, if an experiment has the property of \textbf{Operational Time Symmetry}, this means that if the roles of preparations and measurements in the experiment are swapped, reversing their temporal order, the observed probabilities should remain unchanged. The temporal ordering of events is represented in the causal influence diagram shown in Fig.~\ref{fig:CFcausalinf} for better understanding. Clearly, the notion of time-symmetry adopted here, following \cite{Price2012Retro,leifer2017}, differs from the conventional notion of time-symmetry typically discussed in the physics literature to express the symmetry of dynamics.

A natural question is whether an operational theory should be expected to be \textbf{Operationally Time Symmetric}. The intuitive answer is \emph{No}. This becomes clear if we recall the basic prepare-and–measure setup relevant to a laboratory experiment. Due to the thermodynamic arrow of time, it is natural to demand that future measurements should not affect past preparations, i.e.,  
\begin{equation}
\label{eq:time1}
    p(e|u)= \sum_{f} p(e,f|u,v) = \sum_{f} p(e,f|u,v') \quad \forall v,v'. 
\end{equation}
In contrast, it is not natural to expect the reverse condition, namely that preparations should not influence future measurements, i.e.,  
\begin{equation}
\label{eq:time2}
    p(f|v)= \sum_{e} p(e,f|u,v) = \sum_{e} p(e,f|u',v) \quad \forall u,u', 
\end{equation}

\begin{center}
    \begin{figure}
     \includegraphics[height=50mm, width=85mm,scale=4]{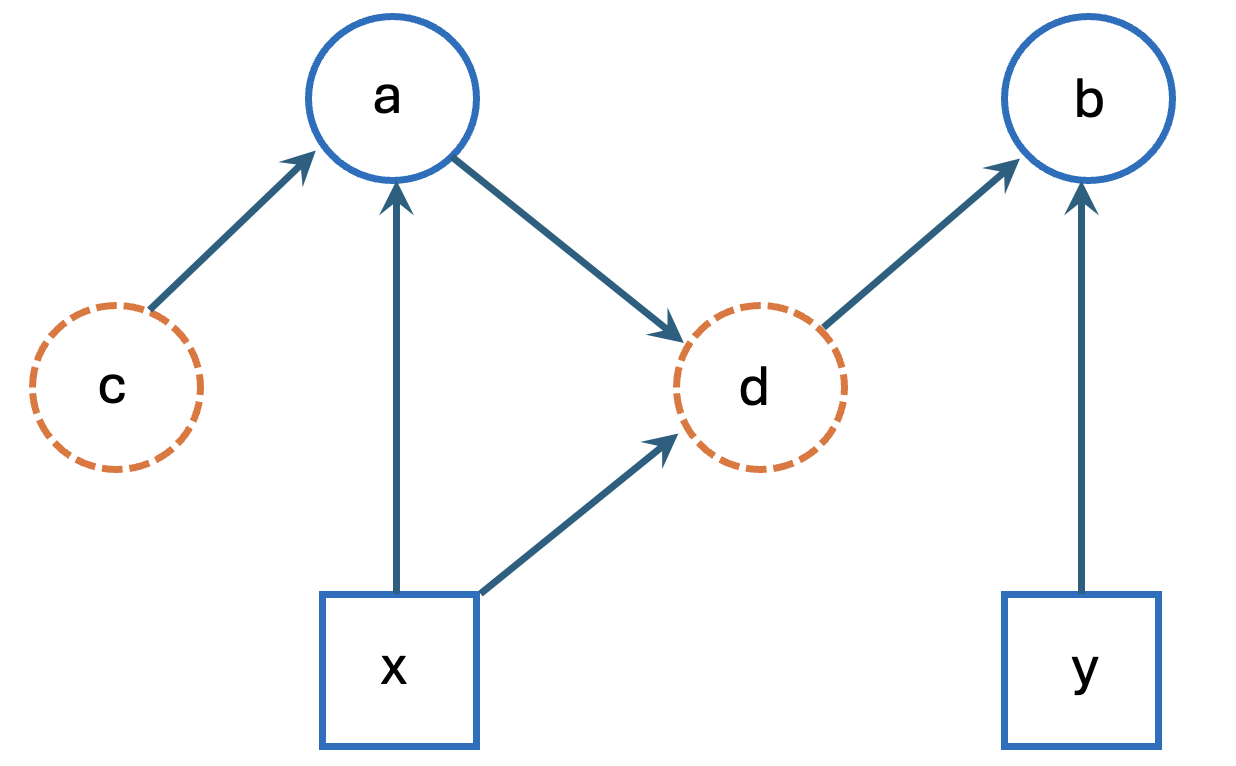}
    \caption{Causal influence diagram for the Causal Time-Symmetric Friendliness scenario. The rectangular entities are those that are in direct control of the observer, while the circles represents events that are not in direct control. The arrows represents the order of the events conveying the cause-effect relationships between the events. Red coloured shapes indicate the Pseudo Events while the blue one are Truly-Observed Events.}
    \label{fig:CFcausalinf}
\end{figure}
\end{center}

The impossibility of satisfying Eq.~\eqref{eq:time2} is the main obstacle in having an \textbf{Operationally Time Symmetric} description of a theory. Nevertheless, without loss of generality, we may restrict our attention to those preparations for which Eq.~\eqref{eq:time2} \emph{does} hold. We refer to the corresponding set of behaviour as \emph{no-signalling in time}, denoted by $\mathcal{NS}_{T}$. It was shown in \cite{leifer2017} that the sector of quantum theory satisfying $\mathcal{NS}_{T}$ is \textbf{Operationally Time Symmetric}. In what follows, while establishing our main results, we restrict to quantum preparations belonging to $\mathcal{NS}_{T}$ so that \textbf{Operational Time Symmetry} is guaranteed.

It is important to note that the definition of \textbf{Operational Time Symmetry} pertains exclusively to events that are both real and operationally accessible. In contrast, the Causal Time-Symmetric Friendliness scenario includes not only these real events but also additional elements referred to as Pseudo Events. Consequently, in such a scenario, only an analogue of \textbf{Operational Time Symmetry} can be meaningfully defined. Motivated by this distinction, we employ this analogue to establish the main results of the paper.

In the Causal Time-Symmetric Friendliness scenario, the pair Charlie-Alice acts before Debbie Bob in forward time, so the former pair naturally plays the role of the preparers while the latter acts as the measurers. In the time reversed description, however, Debbie-Bob acts before Charlie-Alice, thereby interchanging the roles of preparation and measurement. At first sight, one might therefore be tempted to interpret the local ordering at each site as defining an independent prepare-measure scenario, for instance by viewing Charlie as preparing the system and Alice as subsequently measuring it at one site, with an analogous interpretation for Debbie and Bob at the other site. However, such a description of the prepare-and-measure structure is ill-defined. To see this, consider a particular experimental run in which $x = 1$ (or $y = 1$), so that the corresponding event $c$ (or $d$) is erased and effectively treated as if it never occurred. In such a run, there is thus a measurement event without a corresponding preparation event, which shows that interpreting $c$ alone as a preparation is inconsistent. As discussed earlier, the events $c$ and $d$ are Pseudo Events, since they are either erased or ultimately mapped to the outcomes $a$ and $b$, respectively. Consequently, a consistent operational description only allows the pairs $(c,a)$ to be regarded as a preparation and $(d,b)$ as a measurement, rather than splitting these roles between the individual events $c$ and $a$, or $d$ and $b$.

This means the counterpart of the definition of \textbf{Operational Time Symmetry} namely the \textbf{Axiological Time Symmetry} assumption applicable to this scenario does not involve reversing the entire temporal sequence $c \rightarrow a \rightarrow d \rightarrow b$ to $b \rightarrow d \rightarrow a \rightarrow c$. Instead, it reverses the temporal order of the grouped preparation and measurement events, namely $(c, a) \rightarrow (d, b)$, into $(d, b) \rightarrow (c, a)$.
Considering this, a generalisation of \textbf{Operational Time Symmetry} can be proposed when this definition is extended to include the Pseudo Events. This forms the assumption of \textbf{Axiological Time Symmetry}, the final assumption needed for us to generate a full \textbf{Causal Time-Symmetric Friendliness No-Go Theorem}. Before stating the theorem, it is essential to formally define other assumptions and their implications, which are discussed in the following.

\begin{assumption}{\textbf{Axiological Time Symmetry:}}
   If the Operational Time Symmetry holds for Truly-Observed Events in a given scenario then it must also hold for all the events including the Pseudo Events involved in the scenario.
    
    In the taxonomy the Causal Time-Symmetric Friendliness settings this means that if statistics generated by the events corresponding to Alice's and Bob's measurements satisfies Operational Time Symmetry then in order to satisfy Axiological Time Symmetry the following relation must hold,
    \begin{equation}
    p_{\leftarrow}(d,b,c,a|y,x) =  p_{\rightarrow}(c,a,d,b|x,y) \ \forall x,y,a,b,c,d.
    \end{equation} 
\end{assumption}

We emphasise that we neither claim that \textbf{Operational Time Symmetry} holding at the level of Truly-Observed events means it must also hold for all events in scenarios involving Pseudo Events, nor that any underlying description incorporating pseudo events must satisfy \textbf{Axiological Time Symmetry} independently of the theory under consideration. Rather, in this work, we deliberately restrict our attention to cases in which \textbf{Axiological Time Symmetry} holds, and investigate what can be inferred about quantum theory under this assumption. Our aim in this paper is to examine the resulting constraints on observable correlations, assuming that \textbf{Axiological Time Symmetry} and the other stated assumptions hold, and to compare these constraints with quantum predictions. In the present scenario, \textbf{Axiological Time Symmetry} should be understood as an additional structural assumption extending \textbf{Operational Time Symmetry} beyond empirically accessible Truly-Observed Events to include Pseudo Events. While this extension is not compelled by operational considerations alone, we argue that it is a coherent and well motivated assumption within the specific Causal Time-Symmetric Friendliness framework considered here, and we make this restriction explicit while establishing our main results.

One might object the fact that the assumption of \textbf{Axiological Time Symmetry} already presupposes \textbf{Absoluteness of Observed Events} hold, since it is formulated in terms of a joint probability distribution over all events. We argue that although AOE is a prerequisite for defining ATS, this is a very natural and unavoidable pre-assumption rather than an additional hidden commitment. This point can be understood by analogy with Bell's notion of local causality \cite{Bell1964,Wiseman2017}. In deriving a Bell inequality, one typically invokes realism, locality, and freedom of choice. However, locality cannot be meaningfully formulated without first specifying realism, as the notion of locality presupposes the existence of definite underlying events or variables to which locality can apply. In this sense, realism is already implicitly assumed when locality is defined. In an analogous, though conceptually distinct manner, it is still meaningful for us to investigate the effect of the assumption of \textbf{Axiological Time Symmetry}, even though it implicitly assumes \textbf{Absoluteness of Observed Events}. Specifically, in order to define \textbf{Axiological Time Symmetry} in probabilistic terms, one must first assume \textbf{Absoluteness of Observed Events}.
    
\begin{assumption}{\textbf{No Retrocausality:}}
   Any future measurement choice  cannot influence the past outcomes corresponding to the Pseudo Events  and any Truly-observed event in the past cannot be influenced by the future Pseudo Events.

In the present scenario this assumption implies that events corresponding to the outcome $ c(d) $ is independent of the choice of measurement $ x(y) $ and the Truly-Observed Event $a$ is should be independent Pseudo Event $d$. 
\end{assumption}

Here, the terms future and past are used in their standard physical sense - namely, an event $X$ is said to occur in the future (past) of another event $Y$ if it lies within the future (past) light cone of the spacetime region associated with $Y$. Moreover, throughout the paper, statements of the form ``$X$ cannot influence $Y$'', in particular within the \textbf{No Retrocausality} assumption, are to be understood in a strictly operational and statistical sense. By this we mean that appropriate conditional independence relations hold for the relevant probability distributions, for example
\begin{equation}
\begin{split}
&p(c | x,y) = p(c), \\
&p(d | y) = p(d), \\
&p(a | c,x,d) = p(a | c,x). 
\end{split}
\end{equation}

We do not assume a general causal Markov condition associated with directed acyclic graph-based causal models~\cite{pearl2009causality}. Instead, these conditional independences are postulated directly as scenario-specific restrictions motivated by the operational structure of the Causal Time-Symmetric Friendliness setup. The phrase ``cannot influence'' should therefore be read exclusively as a shorthand for these operational conditional independence constraints.

The \textbf{No Retrocausality} assumption emerges by analogy to the idea of retrocausality in hidden variable models \cite{Price2012Retro,Costa_de_Beauregard1976-oa}. Retrocausality is one of the possible interpretations of the violation of the Statistical Independence assumption in hidden variable models - i.e., a way of interpreting the idea that there may be a correlation between the hidden variables $\lambda$ in the model, and the measurement choices \cite{Wharton2020Retrocausality,adlam2024taxonomy}. Given these measurement choices are typically defined as being ``free'', it is specifically the interpretation that the hidden variable $\lambda$ is in some way causally dependent on these measurement choices. This is as opposed to the interpretation of the violation of Statistical Independence as Superdeterminism, which denies this definition of the measurement choices as ``free'', and instead says that they can be causally dependent on the hidden variables, or the interpretation of this violation through counterfactual restriction, where instead simply certain values of hidden variable and certain choices of measurement simply cannot exist, without there necessarily being a causal direction to that prohibition \cite{Hance2024CFRestrictions}. Note that, despite retrocausality implying a causal dependence of the hidden variables on the measurement choices, this does not necessarily imply any backwards-in-time signalling from the point the choices were made - indeed, most modern versions of retrocausality are all-at-once models, where, despite causal dependency, the model is treated as a ``block universe'' structure, where time is just a dimensional parameter, rather than implying any sort of ordered priority.

The retrocausality we appeal to with one of the parts of our \textbf{No Retrocausality} assumption takes a similar form to this retrocausality, except the prohibited relation is a causal dependence of the (potentially-observable) past Pseudo Events on the future measurement choices, rather than the dependence of any hidden variables on the measurement choices. This is as an instance  of \textbf{No Retrocausality} (namely Pseudo Events are independent future measurement choices) where the Pseudo Events act equivalently to hidden variables in the conventional description of retrocausality. The other part of \textbf{No Retrocausality} comes from the idea that a Truly-Observed Past event $a$ cannot also depend on any of the non input variables $d$ that are potentially erased depending on a free choice.

Note that, in no part of the paper do we claim that $c, d$ should be regarded as something that exists independent of the observation that contains intrinsic properties of the system. Rather, we say if a Pseudo Event $c$ or $d$ occurs then it creates accessible information about the system (even if erased) which can in principle influence any future event corresponding to any future measurement outcome. This inspires us to define the final metaphysical assumption as follows:

\begin{assumption}{\textbf{Screening via Pseudo Events:}}
  Every Truly-Observed Event is screened off from all the non-input variables that lies in its causal past by Pseudo Events that its past light cone contains.
  
In the current scenario, for Bob’s observed event, this assumption implies:
    \begin{equation}
       \ \ \ \ \ \ p(b|c,x,a,d,y)=p(b|c,x,d,y) \ \ \forall b,c,x,d,y
    \end{equation}
\end{assumption}

The variables $c$ and $d$, representing the Pseudo Events, impose constraints on Truly-Observed Events strictly weaker than those imposed by hidden variables in conventional no-go theorems. Hidden variables are typically assumed to encode complete information about a physical system and (at least in principle) enable an event-by-event description of its behaviour. In contrast, the \textbf{Screening via Pseudo Events} assumption endows $c$ and $d$ with only limited power, allowing them to influence solely those Truly-Observed Events lying within their future light cones, and screening the future events off from any past events that are already influenced by, or that are already influencing, these Pseudo Events. For example, if the Pseudo Event $c$ exists and is absolute, then the Truly-Observed Event $a$ must always depend on $c$, regardless of the choice of $x$. In particular, for $x=0$ one has $a=c$, while for other values of $x$, even if $c$ were erased and the laboratory restored to its prior state, the effect of any past information on $a$ could still be argued to be screened through the pre-existing value of $c$. In this sense, so long as there is the assumption (implicit in, but weaker than, \textbf{Absoluteness of Observed Events}), that the Pseudo Events $c$ and $d$ are Absolute, then \textbf{Screening via Pseudo Events} emerges as a very natural assumption. We will discuss this implicit assumption further in Section \ref{Subsect:impossAPE}, where we attempt to subdivide \textbf{Absoluteness of Observed Events} to see which parts of it are actually necessary to generate a contradiction with quantum mechanics. While the implications of relaxing \textbf{Absoluteness of Observed Events} as a whole have been considered in literature \cite{Adlam2024AOE,walleghem2024,Wiseman2023thoughtfullocal}, we believe we are the first to consider the effect of sub-assumptions of \textbf{Absoluteness of Observed Events}.

\begin{thm}[Causal Time-Symmetric Friendliness No-Go Theorem]
\label{Causal_Friendliness_theorem}
Consider the scenario we give in Section~\ref{Subsect:NewScenario}, with Charlie and Alice acting before Debbie and Bob. 
Let $a,b,c,d\in\{\pm1\}$ denote the outcomes of Alice, Bob, Charlie, and Debbie, respectively, and let $x,y\in\{0,1\}$ denote the measurement settings of Alice and Bob. 
If the operational statistics admits \textbf{ Absoluteness of Observed Events}, \textbf{Axiological Time Symmetry}, \textbf{No Retrocausality}, and \textbf{Screening via Pseudo Events} , the observed correlations satisfy the inequality:
\begin{equation}
\label{chsh}
    \Big|\,\langle A_{0} B_{0} \rangle +\langle A_{0} B_{1} \rangle +\langle A_{1} B_{0} \rangle -\langle A_{1} B_{1} \rangle \,\Big| \;\le 2\,,
\end{equation}
where $\langle A_{x} B_{y} \rangle = \sum_{a,b=\pm1} ab\,p(a,b|x,y)$ are the joint expectation values for the measurements, conditioned on the inputs $(x,y)$.
\end{thm}
{\emph{Sketch of the proof:} Starting from \textbf{Absoluteness of Observed Events}, we assume the existence of a joint probability distribution \(p(c,a,d,b|x,y)\), from which the observed statistics are obtained by marginalisation. The assumption of \textbf{No Retrocausality} and \textbf{Axiological Time-Symmetry} further implies $p(c,d|x,y) = p(c,d)$, $p(a|c,d,x,y ) = p(a|c,x,d)$, and $p(b | y, c, x, d) = p(b | y, c, d)$. Further, \textbf{Screening via Pseudo Events} implies relations of the form \(p(b|a,c,d,x,y)=p(b|c,d,y)\). Combining these constraints with Axiological Time Symmetry, one obtains that the joint distribution admits the factorisation
\begin{equation}
 p(a,b|x,y) = \sum_{c,d} p(c,d)\, p(a|x,c)\, p(b|y,c,d).
\end{equation}
The bound in Eq.~\eqref{chsh} then follows directly from this decomposition. The complete proof of this theorem is in Appendix~\ref{App:ProofCFTheorem}.}

\subsection{Relation to the causal Markov condition}
It is useful to clarify how the assumptions employed in this work relate to the standard causal Markov condition used in directed acyclic graph-based causal models \cite{pearl2009causality}. For the forward time ordering of the Causal Time-Symmetric Friendliness scenario, a natural causal structure would suggest the Markov factorisation

\begin{equation}
\begin{split}
p&(c,a,d,b | x,y) =\\ &p(c)\,p(a | c,x)\,p(d | c,a)\,p(b | c,a,d,y),
\end{split}
\end{equation}

from which all conditional independence relations would normally follow.
In this work, however, we do not assume the causal Markov condition. Instead, we impose a restricted set of operationally-motivated conditional independence constraints, encoded in the assumptions of \textbf{No Retrocausality} and \textbf{Screening via Pseudo Events}. As a consequence, some of the conditional independences implied by the causal Markov condition are retained, while others are do not appear in our factorisation. Specifically, \textbf{No Retrocausality} enforces that future measurement choices influence neither past Pseudo Events nor past Truly-Observed Events. This retains Markov-type independences such as the independence of $c$ from $x,y$, and the independence of $a$ from $y$ once its past variables are fixed. \textbf{Screening via Pseudo Events} further ensures that future Truly-Observed Events are screened off from earlier events by appropriate Pseudo Events, thereby retaining the Markov style screening of $b$ by $c,d,y$.

Crucially, we do not require that conditioning on immediate causal parents always removes all further dependencies. In particular, we allow the Pseudo Event $d$ to depend on either Alice's measurement choice nor her outcome, and we do not impose the minimal parent screening characteristic of the causal Markov condition. These dropped independences are essential for accommodating Pseudo Events, and for distinguishing our framework from standard hidden variable models or directed acyclic graph-based causal models. The relation between the conditional independences implied by the causal Markov condition and those assumed in this work is summarised in Table~\ref{tab:markov_comparison}. Overall, our framework abandons the global causal Markov assumption and replaces it with a controlled set of conditional independence relations justified directly by \textbf{No Retrocausality} and \textbf{Screening via Pseudo Events}, which are sufficient for deriving the Causal Time-symmetric Friendliness inequality.

 \section{Role of Pseudo Events in the Causal Time-Symmetric Friendliness Paradox}\label{Sect:CFPWeakenedAOE}

From the analysis in the previous section, a natural question arises: does the violation of Eq.~\eqref{chsh} imply that the observed events are not absolute (even the Pseudo Events), but instead subjective to the observers? To address this, we first examine to what extent the \textbf{Absoluteness of Observed Events} assumption is truly necessary, in addition to \textbf{Axiological Time Symmetry}, \textbf{No Retrocausality}, and \textbf{Screening via Pseudo Events}, for establishing the Causal Time-Symmetric Friendliness no-go theorem. 
In this context, we attempt to supersede the \textbf{Absoluteness of Observed Events} assumption by drawing an analogy with the replacement of the ``wavefunctions are \emph{$\psi$-ontic}'' assumption used in Price's no-go result \cite{Price2012Retro}, by the $\lambda$-mediation assumption (discussed later in this section), which Leifer and Pusey used to establish a no-go theorem~\cite{leifer2017} about time symmetry and retrocausal ontology. Price's no-go theorem \cite{Price2012Retro} suggests that any time-symmetric ontology of quantum theory must entail retrocausality. His argument proceeds by showing that the following three assumptions, when taken together, contradict quantum theory:
\begin{itemize}
    \item \textbf{No Ontic Retrocausality}: \emph{Future measurement choices cannot influence the distribution of past ontic outcomes. In the forward-time ordering this gives the most general causal factorisation}
  \begin{equation}
       p(\lambda,a,b | x,y)
    \;=\; p(\lambda)\,p(a | x,\lambda)\,p(b | a,\lambda,y).
  \end{equation}
   Note this is defined analogously to condition \textbf{No Retrocausality} given earlier in this paper, but applied to hidden variable $\lambda$ rather than Pseudo Events $c$ and $d$;
    \item \textbf{Ontological Time-Symmetry}: \emph{The ontological probabilities for the forward ordering equal those of the appropriate time-reversed ordering, i.e.}
  \begin{equation}
    p_{\rightarrow}(\lambda,a,b| x,y) \;=\; p_{\leftarrow}(b,a,\lambda| y,x).
  \end{equation}
    This is the the ontological counterpart of the \textbf{Axiological Time Symmetry} assumption, that is applied to probabilities associated with $\lambda$ instead of the Pseudo Events;
    \item \textbf{Reproducibility with a \emph{$\psi$-ontic} ontological model}:  \emph{Compatibility of quantum statistics with the assumption that ``No two distinct quantum states have overlapping ontic state distributions'', in the sense of Harrigan and Spekkens’s definition of $\psi$-ontic~\cite{Harrigan2010Nonlocality}, formulated within Spekkens’s ontological models framework~\cite{Spekkens2005Contextuality}.}
\end{itemize}

\begin{table*}[t]
\centering
\begin{tabular}{lll}
\hline
Conditional independence & Causal Markov \ \ \ & \ \ \  Status here \\
\hline
$p(c | x,y)=p(c)$ & \ \ Implied & \ \  Retained by \textbf{No Retrocausality} \\
$p(a | c,x,d,y)=p(a | c,x,d)$  & \ \  Implied & \ \  Retained by Lemma \ref{lem:bob} \\
$p(b | a,c,d,x,y)=p(b | c,d,y)$ & \ \  Not implied & \ \  Retained by \textbf{Screening via Pseudo Events} and Lemma \ref{lem:bob} \\
$p(d | c,a,x)=p(d | c,a)$ & \ \ Implied & \ \  Does not appear \\
$p(b | c,a,d,y,x)=p(b | d,y)$ & \ \ Implied & \ \  Does not appear\\
\hline
\end{tabular}
\caption{Comparison between conditional independences implied by the causal Markov condition and those employed in this work.}
\label{tab:markov_comparison}
\end{table*}

Note however that ontological models are only a small subset of possible hidden variable models. Further, Ref.~\cite{tezzin2025ontological,tezzin2025det} recently showed that ontological models have trouble reproducing key aspects of the predictions of quantum mechanics, even before we consider the subset which are $\psi$-ontic. Therefore, reproducibility with a $\psi$-ontic model is not a natural assumption to use when proving the fact that time-symmetric interpretation of quantum theory would be retrocausal. The status of the state, whether it is ontic or epistemic (or potentially both \cite{Hance2022Simultaneously}), should not come into the proof. In order to address this problem, and to prove Price's result without attributing any ontological status to the state, Leifer and Pusey \cite{leifer2017} showed that combining Price's assumptions of No Retrocausality and Time-symmetry with another assumption, $\lambda$-mediation, leads to a contradiction with quantum theory \cite{leifer2017}.

\subsection{Replacing \texorpdfstring{$\psi$-\emph{ontic}}{psi-ontic} reproducibility with \texorpdfstring{$\lambda$}{lambda}-mediation}\label{Subsect:lambdafrompsi}

In the ontological models framework, where the quantum state (wavefunction) is represented by a probability distribution over a space of underlying 'ontic' states (see~\cite{Leifer2014} for details), models are classified as either \emph{$\psi$-ontic} or \emph{$\psi$-epistemic}. In \emph{$\psi$-ontic} models, distinct quantum states correspond to non-overlapping distributions over this ontic state space. \emph{$\psi$-epistemic} models instead allow quantum states to overlap on this space. The extra assumptions imposed on the possible ontological counterpart of the quantum state often restrict the statistics that such state can generate.

Price's result~\cite{Price2012Retro} on the time-symmetric ontology of quantum theory lacks generality, since his argument relies explicitly on assuming reality is \emph{$\psi$-ontic}. While \emph{$\psi$-epistemic} models are themself subject to significant criticisms~\cite{Pusey2012_PBR,Aaronson_PRA13,hardy_PI,barrettepsitemic14,Cyril_epistemic,montina15,Leifer13,leifer14}, even together these two categories leave out all models of the wavefunction which do not follow the ontological models framework. This suggests that, to properly investigate time-symmetric ontologies, we should avoid limiting ourselves to one or the other categories of ontological model. 

In this direction, Leifer and Pusey~\cite{leifer2017} proposed a weaker requirement, known as \emph{$\lambda$-mediation}, according to which all correlations between preparations and measurement outcomes are mediated via the underlying ontic state~$\lambda$ in an Ontological Model of the scenario. Formally, this means that the probability of an outcome $b$ for a arbitrary preparation $P$ and a measurement $M$ can be written as
\begin{equation}
  p(b|P,M) = \sum_{\lambda} p(\lambda|P)\, p(b|M,\lambda),
\end{equation}
without requiring that the ontic distributions $p(\lambda|P)$ corresponding to different quantum states be either disjoint or overlapping. Thus, $\lambda$-mediation relaxes Price's $\psi$-ontic reproducibility assumption by removing the requirement of attributing a particular status to quantum state, while still preserving the idea that $\lambda$ screens off the correlations between preparation and measurement.

We here ask whether we can relax the \textbf{Absoluteness of Observed Events} assumption to a weaker assumption in the same way as Leifer and Pusey replace Price's $\psi$-ontic assumption with $\lambda$-mediation \cite{leifer2017}. Given the \textbf{Absoluteness of Observed Events} assumption seems similar to the $\psi$-ontic assumption (albeit imposed on the operational and observational level instead of the ontological/hidden-variable level \footnote{Quantum states encode probabilities for all possible events, typically corresponding to measurement outcomes. In a 
$\psi$-ontic model, distinct quantum states correspond to ontic states with disjoint support in the underlying ontic state space. Analogously, the assumption of Absoluteness of Observed Events asserts that distinct observed events correspond to disjoint support in the relevant probability space. Thus, although the assumptions are structurally and physically different, Absoluteness of Observed Events plays a role analogous to $\psi$-onticness when one is concerned with the definiteness and global consistency of an underlying description.}), we argue that the answer is ``yes''. To show this is the case, we first examine formally the relation between the \textbf{Absoluteness of Observed Events} assumption, and the $\psi$-ontic criterion for ontological models. We then show that the \textbf{Absoluteness of Observed Events} assumption can be relaxed, to an operational equivalent of $\lambda$-mediation, which is weaker than AOE, while still allowing us to form Local Friendliness-style inequalities which quantum theory violates.

\subsection{Can we weaken the \textbf{Absoluteness of Observed Events} assumption?}\label{Subsect:WeakeningAOE}

The logical structure underlying our weakening of the \textbf{Absoluteness of Observed Events} assumption closely follows Leifer and Pusey's weakening of Price's $\psi$-ontic assumption to $\lambda$-mediation \cite{leifer2017} - both replacing a strong \emph{absoluteness} assumption ($\psi$-ontic reproducibility or \textbf{Absoluteness of Observed Events}) with a weaker one (ontic or operational mediation). The key similarity between the arguments in \cite{leifer2017} and ours lies in the fact that the independence of the mediator variables from the choice of measurement settings are not assumed but rather derived from other underlying assumptions.

Let the Causal Time-Symmetric Friendliness scenario be as in Sec.~\ref{Sect:CFP}. We replace \textbf{Absoluteness of Observed Events} (existence of the full joint $p(c,a,d,b | x,y)$) with the following weaker, purely operational assumption:
\begin{definition}[Operational Pseudo Event Mediation ]
\label{thm:opem}
A theory admits \textbf{Operational Pseudo Event Mediation} if all of its statistics corresponding any arbitrary time-ordered prepare-measure scenario satisfies the following two assumptions.
\begin{enumerate}
  \item \textbf{Existence of marginals:} For all measurements there exist
  well-defined operational marginals for the joint probabilities of the Pseudo Events and conditional response functions for the Truly-Observed Events that are independent of all the measurements settings excepts its immediate cause, describing how the outcomes of Truly-Observed Events depend on Pseudo Events and choice of measurements (The assumption does not require that these response functions and the joint probability distribution for the Pseudo Events arise from a single global probability distribution.)\\
  For Causal Time-Symmetric Friendliness scenario this implies that for all settings $(x,y)$ there exist operational marginals $p(c,d | x,y)$ and conditional response functions $p(a | x,c)$ and $p(b | y,c,d)$.
  \item \textbf{Operational mediation:} All the Pseudo Events act as mediators between past and future events, such that once they are specified, one Truly-Observed Event provides no information about a future Truly-Observed Event beyond that already provided by the Pseudo Events in its past and the measurement that is its immediate cause.\\ 
  The specific version of this assumption, tailored to the present scenario, can be written technically as follows:
  \begin{equation}
    p(b | a,c,d,x,y)=p(b | c,d,y),
    \label{eq:opem_spe}
  \end{equation}
 for all arguments.
\end{enumerate}
\end{definition}

The above definition can be summarised as the following implication,
\begin{align}
    &\textbf{Existence of marginals} +\textbf{Operational mediation} \nonumber \\
& \ \implies \textbf{Operational Pseudo Event Mediation}
\end{align}
With these assumptions in place, we establish the following result.

\begin{thm}\label{thm:2}
 Any theory that admits \textbf{Operational Pseudo Event Mediation}, together with the \textbf{Axiological Time Symmetry} and \textbf{No Retrocausality} assumptions (as defined in Sec.~\ref{Sect:CFP}), yields the following factorisation of the operational probabilities
\begin{equation}
\label{eq:opem_factorisation}
  p(a,b | x,y) \;=\; \sum_{c,d} p(c,d)\,p(a | x,c)\,p(b| y,c,d),
\end{equation}
and hence the Causal Time-Symmetric Friendliness bound $|S|\le 2$ follows.   
\end{thm}

{To provide some intuition for Theorem \ref{thm:2}, we briefly outline the main idea of the proof, with full details given in Appendix~\ref{App:ProofTheorem2}. From Operational Pseudo Event Mediation, we use the existence of operational marginals $p(c,d|x,y)$ and response functions $p(a|x,c)$, $p(b|y,c,d)$ to construct a joint distribution $q(a,b,c,d|x,y)$. \textbf{Axiological Time Symmetry} then implies that this construction is invariant under the interchange $(c,a,x)\leftrightarrow(d,b,y)$, yielding $p(c,d|x,y)=p(c,d|y,x)$. Combining this with \textbf{No Retrocausality} enforces $p(c,d|x,y)=p(c,d)$. Substituting this into the constructed distribution yields the factorisation of Eq.\eqref{eq:opem_factorisation} and so the bound $|S|\le 2$ follows directly.}

We emphasize that the motivation for introducing \textbf{Existence of Marginals} and\textbf{Operational Mediation} is not merely to weaken assumptions mathematically, but to identify the minimal operational structure required to derive the Causal Time-Symmetric Friendliness inequalities. To illustrate the distinction between \textbf{Absoluteness of Observed Events} and the weaker assumptions of \textbf{Existence of Marginals} and \textbf{Existence of Marginals}, consider a prepare-measure scenario in which certain intermediate events may be unitarily erased by a later intervention. In such a case, the operationally accessible statistics consist of well defined marginals, such as $p(c,d | x,y)$, together with response functions $p(a | x,c)$ and $p(b | y,c,d)$ associated with Truly-Observed Events. However, when the intermediate events are erased, there is no operational procedure that allows one to access or verify the existence of a single joint probability distribution over all events in a given run. Assuming the existence of such a joint distribution therefore goes beyond what is operationally warranted. This motivates the assumption of \textbf{Existence of Marginals}.

At the same time, it remains operationally natural to require that, once the relevant pseudo events in the causal past are specified, future observed outcomes are screened off from additional variables. In this sense, pseudo events can be viewed as revealing intrinsic properties of the system that constrain what future events can or cannot occur. This requirement is captured by \textbf{Operational Mediation} and expresses a minimal notion of causal screening that does not rely on underlying ontic states or the full causal Markov condition. The following toy example explicitly demonstrates the legitimacy of a situation where \textbf{Existence of Marginals} and \textbf{Operational Mediation} can hold while \textbf{Absoluteness of Observed Events} fails.

\emph{Example: \textbf{Existence of Marginals} and \textbf{Operational Marginals} without \textbf{Absoluteness of Observed Events}.}
 Consider the pseudo events $c,d\in\{0,1\}$, and Truly-Observed Events $a,b\in\{0,1\}$, and inputs $x,y\in\{0,1\}$. Let the pseudo events be distributed as $p(c,d)=1/4$, independent of $x,y$. Alice’s response functions are defined by $p(a=c | x=0,c)=1$ which depends on $c$ and $p(a=0 | x=1,c)=p(a=1 | x=1,c)=1/2$, independent of $c$. Bob’s response functions are defined analogously by $d$-dependent probability $p(b=d | y=0,d)=1$ and $p(b=0 | y=1,d)=p(b=1 | y=1,d)=1/2$, independent of $d$. 

All operationally relevant marginals $p(c,d)$, $p(a | x,c)$, and $p(b | y,d)$ are well defined, and the pseudo events screen off future outcomes, so that $p(b | a,c,d,x,y)=p(b | y,d)$ and $p(a | b,c,d,x,y)=p(a | x,c)$, satisfying EOM and OM. However, because the outcomes associated with $x=0$ and $x=1$ (and similarly $y=0$ and $y=1$) cannot be jointly accessed in a single run due to the possibility of unitary erasure of pseudo events, the operational data do not warrant the existence of a single joint distribution $p(c,a,d,b | x,y)$. For instance, the relevant marginals can come from a joint quasi-probability distribution \cite{wignerquasi,kirkwood}. Thus in this toy example \textbf{Absoluteness of Observed Events} is not satisfied, showing that \textbf{Operational Pseudo Event Mediation} is strictly weaker than it.

It is important to note that the assumption \textbf{Operational mediation} is stronger than the \textbf{Screening via Pseudo Events} assumption used in the proof of the Causal Time-Symmetric Friendliness no-go theorem. In other words, \textbf{Screening via Pseudo Events} can be viewed as a specific instance of \textbf{Operational mediation}, where the dependence of the event $b$ on the past input choice $x$ is not relaxed. In the proof of Causal Time-Symmetric Friendliness no-go theorem, however, the independence of $b$ from $x$ is derived from the joint assumptions of \textbf{Absoluteness of Observed Events}, \textbf{Axiological Time Symmetry}, and \textbf{No Retrocausality}. In contrast, in the present analysis, we are not permitted to assume \textbf{Absoluteness of Observed Events} at the outset, and consequently, the independence of $b$ from $x$ is incorporated directly into the relatively stronger assumption of \textbf{Operational mediation} itself.

Importantly, although \textbf{Existence of marginals} by itself does not require the existing marginals to originate from a single global joint probability distribution, imposing the requirement that those existing marginals given by \textbf{Existence of marginals} originate from a single global joint probability distribution leads to \textbf{Absoluteness of Observed Events}. Then we clearly have,



\begin{align}
   \textbf{Absoluteness of } & \textbf{Observed Events} \nonumber\\
   \implies& \textbf{Existence of marginals} 
\end{align}

However, the reverse implication may not hold in general as \textbf{Existence of marginals} does not necessarily involve a single factorizable joint probability distribution.

\subsection{Impossibility of discarding the Absoluteness of Pseudo Events}\label{Subsect:impossAPE}

Above, we showed that we can weaken the \textbf{Absoluteness of Observed Events} assumption to \textbf{Existence of marginals}, and still obtain a contradiction with quantum mechanics, so long as we strengthen \textbf{Screening via Pseudo Events} into \textbf{Operational mediation}. We now argue that we can split \textbf{Existence of marginals} up further, into \textbf{Absoluteness of Truly-Observed Events} , and \textbf{Absoluteness of Pseudo Events} defined respectively as,

\begin{definition}[\textbf{Absoluteness of Truly-Observed Events}]
    A theory satisfies Absoluteness of Truly-Observed Events if all observed events at a particular point on the space-time region characterised by measurement outcomes $a,b$ have definite values regardless irrespective of any event that occurred at any other point in the space time. This means that the joint probability $p(a,b|x,y)$ exists.
\end{definition}

\begin{definition}[\textbf{Absoluteness of Pseudo Events}]
        A theory satisfies Absoluteness of Pseudo Events if all observed events at a particular point on the space-time region characterised by measurement outcomes $c,d$ have definite values regardless irrespective of any event that occurred at any other point in the space time. This means that the joint probability $p(c,d|x,y)$ exists.
\end{definition}

We can summarise the connection between the assumptions discussed in Sec.~\ref{Subsect:impossAPE} and Sec.~\ref{Subsect:WeakeningAOE} as follows:
\begin{align}
\label{eq:spliting}
  & \textbf{Existence of marginals}+ \textbf{No Retrocausality} \nonumber \\ 
  &\implies \textbf{Absoluteness of Truly-Observed Events}  \nonumber \\
  & \ \ \ \ \ \ \ \ \ \ \ \ \ \ \ \ + \textbf{Absoluteness of Pseudo Events}   
\end{align}

The implication in Eq. \eqref{eq:spliting} is due to the fact that the assumption \textbf{Existence of marginals} already assumes that $p(c,d|x,y)$ exists ( i.e.,\textbf{Absoluteness of Pseudo Events}). Moreover, the second part of \textbf{Existence of marginals} is that $p(a | x,c)$ and $p(b | y,c,d)$ exists. Then together this implies that there exists a probability $p(a,b|x,y)$ as
\begin{eqnarray}
     p(a,b|x,y)&&=\sum_{c,d}p(c,d|x,y)p(a | x,c)p(b | y,c,d) \nonumber \\
     &&=\sum_{c,d}p(c,d)p(a | x,c,d)p(b | y,c,d)
\end{eqnarray}
which is what \textbf{Absoluteness of Truly-Observed Events} is. The second line of the above equation is due to \textbf{No Retrocausality}. We do not claim that this particular implication will hold generally, but arises from the particular way events, Pseudo Events, and operational accessibility are defined in this Causal Time-Symmetric Friendliness scenario. This is because \textbf{No Retrocausality} is a causal assumption, whereas \textbf{Absoluteness of Truly-Observed Events} and \textbf{Absoluteness of Pseudo Events} are not themselves causal statements, and no general implication between them is expected outside the present framework.

Having shown in Section~\ref{Subsect:WeakeningAOE} that we can remove the \textbf{Existence of Joint Probability} assumption while still being able to generate an inequality which quantum mechanics can violate, we now wish to consider whether we can further remove assumption, ideally to see if \textbf{Absoluteness of Truly-Observed Events} alone is enough to produce an inequality which Quantum Mechanics violates. This question is motivated by \textit{prima facie} \textbf{Absoluteness of Pseudo Events} being violated by standard quantum mechanics, given quantum mechanics describes Charlie's (Debbie's) lab as being in a superposition from Alice's (Bob's) perspective until they open the lab, and given Alice (Bob) can unitarily undo Charlie's (Debbie's) measurement - which requires the measurement itself to have been unitary according to quantum mechanics in the first place. By the term ``standard quantum mechanics'', here we mean the description of quantum mechanics that follows from, or which is mathematically equivalent to, the Copenhagen interpretation \cite{solvay}.

To answer this question, let us consider the case without \textbf{Absoluteness of Pseudo Events}, where we have the following factorised form of observed probabilities, 
\begin{equation}
    p(a,b | x,y)=\sum_{c,d} p(c,d | x,y)\,p(a | x,c)\,p(b| y,c,d).
\end{equation}
which gives,
\begin{equation}
\label{prob:depend}
\begin{split}
&S =\\
&\sum_{c,d}\Big[ p(c,d | 0,0) A_0^{(c)}B_0^{(c,d)}+ p(c,d | 0,1)\,A_0^{(c)}B_1^{(c,d)}\\
& + p(c,d | 1,0) A_1^{(c)}B_0^{(c,d)}- p(c,d | 1,1)\,A_1^{(c)}B_1^{(c,d)}\Big],
\end{split}
\end{equation}

We can no longer factor this as a single convex average $\sum_{c,d} p(c,d)\,S_{c,d}$ as done while establishing inequality Eq.~\eqref{chsh}. The pointwise bound $|S^{c,d}|\leq 2$ for fixed $c,d$ does not, by itself, constrain mixtures with different context-dependent weights. As a result, the overall value of $S$ derived without assuming \textbf{Absoluteness of Pseudo Events} may even exceed the quantum bound, undermining the no-go theorem.

To prove this, consider the specific case where $c,d\in\{0,1\}$ and $A_x^{(c)},B_y^{(c,d)}\in\{\pm1\}$ the deterministic responses and set
\begin{align}
p(c,d| x,y) = \begin{cases}
1 & \text{if } (c,d)=(0,0)\\
& \text{ and } (x,y)=(0,0),\\[2pt]
1 & \text{if } (c,d)=(0,1)\\
& \text{ and } (x,y)=(0,1),\\[2pt]
1 & \text{if } (c,d)=(1,0)\\
&\text{ and } (x,y)=(1,0),\\[2pt]
1 & \text{if } (c,d)=(1,1)\\
&\text{ and } (x,y)=(1,1),\\[2pt]
0 & \text{otherwise.}
\end{cases}
\end{align}
Now if we choose responses such that $A_0^{(0)}=A_1^{(1)}=B_0^{(0,0)}=B_1^{(0,1)}=B_0^{(1,0)}=+1$ and $B_1^{(1,1)}=-1$, then,
\begin{equation}
  \langle A_0B_0\rangle=\langle A_0B_1\rangle= \langle A_1B_0\rangle=+1, \langle A_1B_1\rangle=-1. 
\end{equation}
which implies \(S=1+1+1-(-1)=4\), violating not only the bound provided in Eq.~\eqref{chsh} but also the quantum bound $S_{Q}=2\sqrt{2}$, reaching the Boxworld bound $S_{B}=4$ \cite{Popescu1994}. This shows that without mediator independence (derived by combining \textbf{Axiological Time Symmetry}+\textbf{No Retrocausality} (Lemma \ref{lem:deb}) with \textbf{Absoluteness of Observed Events}), it is not possible to establish the no-go result.

\section{Discussions and outlook}\label{Sect:Discussions}
\subsection{Comparing to Related Results}\label{Subsect:Comparisons}


Recent works have connected Local Friendliness–type no-go theorems to both forms of contextuality, namely preparation contextuality~\cite{Catani_25} and Kochen–Specker contextuality~\cite{Walleghem2025}. These works demonstrate that different manifestations of quantum correlations can be used to establish the incompatibility of quantum theory with observer-independent events. This relationship can be understood in terms of a hierarchy of correlations that each no-go theorem employs.

In a prepare-and-measure scenario the connection between preparation contextuality and \textbf{Absoluteness of Observed Events} is formulated through what they refer to as the  \emph{Operational Friendliness} no-go theorem. Crucially, Leifer and Pusey proved that, in a timelike prepare–and–measure scenario, imposing \emph{Time Symmetry} together with \emph{No Retrocausality} and \emph{$\lambda$-mediation} yields the mediator/epistemic independence condition (e.g., $p(\lambda|x)=p(\lambda)$ for the no-signalling preparations) usually assumed in preparation-noncontextuality arguments \cite{leifer2017}. Consequently these assumptions together imply a timelike Bell-CHSH constraint which quantum theory violates. In contrast, the Operational Friendliness no-go theorem derives this contradiction with quantum theory by combining \textbf{Absoluteness of Observed Events} with \emph{Operational Agency}- a stronger version of operational equivalence used in the discussion of Spekkens' contextuality\cite{Spekkens2005Contextuality,Catani_25,Walleghem2025}. By using similar operational correlations involving time-ordered measurements, we employ \textbf{Absoluteness of Observed Events} to establish our Causal Time-Symmetric Friendliness no-go result through a different yet equally powerful route. We therefore view our Causal Time-Symmetric Friendliness no-go theorem as complementary to the Operational Friendliness no-go theorem. However, the two are not hierarchically ordered because Time Symmetry and preparation-noncontextuality are logically different principles and can be independent in other scenarios. 

Note that it is worth investigating the difference between the Causal Time-Symmetric Friendliness constraints with the generalised-noncontextuality constraints for the sequential measurement scenario. However, for the bipartite scenario we discuss here, where each party can choose between two dichotomic measurements (the $2222$-Bell scenario), noncontextuality constraints provide the same Bell-CHSH-like inequality. To determine whether our constraints are weaker or stronger than those imposed by noncontextuality, one would need to consider scenarios beyond the $2222$-Bell scenario, in which a Causal Time-Symmetric Friendliness inequality can be violated while inequalities derived from Kochen–Specker contextuality~\cite{Walleghem2025} or preparation contextuality~\cite{Catani_25} cannot, and vice versa. We leave the derivation of such Causal Time-Symmetric Friendliness inequalities beyond the $2222$-Bell scenario for future work.

\subsection{Implications of the findings for interpretational issues}\label{Subsect:Implications}

It is worth examining how the Absoluteness assumptions we discuss above relate to two prominent realist interpretations of quantum theory: the de~Broglie–Bohm (dBB) theory \cite{solvay,Bohm1,Bohm2}, and Everett's Many–Worlds interpretation (MWI) \cite{Everett}.

In the dBB framework, every system is associated with a set of ontic variables (the system's exact position) that evolves under a guiding equation. These variables assign definite outcomes to all measurement events, so both the Truly-Observed Events and the Pseudo Events possess ontic definiteness in every run. In the terminology of this work, dBB explicitly maintains \textbf{Absoluteness of Truly-Observed Events} and \textbf{Absoluteness of Pseudo Events}. However, as the quantum potential is explicitly nonlocal (allowing measurement results to depend on factors outside the past light-cone of the measurement), in the \textbf{Time Symmetric} sector \footnote{In general, dBB theory, like standard quantum theory, does not satisfy Operational Time Symmetry. However, one may consistently restrict attention to a time-symmetric sector of the theory \cite{pirsa_PIRSA:08090068,Schmidt_2019dBBTimeSymmetric}.}, dBB circumvents our no-go results by abandoning the condition of \textbf{Screening via Pseudo Events}.

On the other hand, the Many-Worlds Interpretation (MWI) fundamentally rejects the very notion of single, absolute outcomes. In the MWI, a measurement does not yield a unique, classical result. Rather, it induces a branching of the global wavefunction into dynamically independent components, each realising a distinct record of the possible outcomes. Consequently, the MWI explicitly denies the \textbf{Absoluteness of Observed Events} assumption on which our no-go theorem is based. From the perspective of an MWI adherent, the apparent incompatibility between quantum predictions and the Causal Time-Symmetric Friendliness inequality simply reflects the fact that the \emph{observed events} in any given branch represent only one relative component of a larger, branch-dependent reality. In this view, the violation of AOE at the operational level arises naturally from the multiplicity of branching records in the universal wavefunction, thereby making the Causal Time-Symmetric Friendliness no-go result completely compatible with the logic used in MWI.

\subsection{Summary and future directions}
In this paper, building on recent results, we introduced the \emph{Causal Time-Symmetric Friendliness Paradox}, a time-ordered analogue of Bong \etal's Local Friendliness no-go theorem. To generate this no-go theorem, we replace the usual locality assumption with \emph{Axiological Time Symmetry} (ATS), showing that, when combined with \emph{Absoluteness of Observed Events} (AOE), \emph{No Retrocausality} (NRC), and \emph{Screening via Pseudo Events} (SPE), we obtain a causal Time-Symmetric Friendliness inequality. We then showed that quantum mechanics violates this inequality, demonstrating quantum statistics must be incompatible with at least one of these assumptions. To probe which assumption might be incompatible, we then examine whether AOE in its entirety is essential for this no-go result. We propose a weaker, operational form of AOE that still leads to the inequality that quantum mechanics violates. This result shows that even under relaxed assumptions, quantum theory resists reconciliation with classical notions of absolute events, reinforcing the foundational significance of Wigner’s Friend-type paradoxes in timelike scenarios. However, we finally showed that trying to remove the absoluteness of the individual Pseudo Events $c$ and $d$ undermines our ability to arrive at a contradiction with quantum mechanics, allowing us to not only surpass the maximum quantum value, but reach the maximum possible algebraic value for the correlators for the scenario. This illustrates the crucial role played by the Absoluteness of Pseudo Events in establishing the no-go result.

Our findings opens several directions for the future work, such as extending the weakening the Absoluteness of Observed Events to derive other no-go theorems based on different set of assumptions (e.g. macrorealism \cite{lgi}), as well as exploring experimentally accessible scenarios that could test genuine Causal Time-Symmetric Friendliness violations. A further direction is to relate the present framework to studies of the quantum-to-classical transition. Since the framework of weakening of absoluteness of events specifies when events acquire effective classicality, it may offer an operational perspective on when and how definite outcomes emerge through decoherence \cite{decoherence1,decoherence2} and coarse-grained measurements \cite{coursegrainedKofler}.

\textit{Acknowledgements -} SM thanks Alok Kumar Pan for many insightful discussions on Wigner’s Friend paradoxes over the past years. JRH acknowledges support from a Royal Society Research Grant (RG/R1/251590), and from their EPSRC Quantum Technologies Career Acceleration Fellowship (UKRI1217).

\appendix

\section{Proof of the Causal Time-Symmetric Friendliness No-Go Theorem}\label{App:ProofCFTheorem}

\textbf{Causal factorisation admitting \textbf{No Retrocausality}:} In the forward time order, Charlie measures his quantum state first (obtaining $c$), then Alice chooses $x$ and so sets a value for $a$, then Debbie measures her quantum state (obtaining $d$), and finally Bob chooses $y$ and so sets a value for $b$. 
Thus, the most general causal factorisation for this scenario that admits \textbf{No Retrocausality} can be written as,
\begin{equation}
\label{nrc1}
\begin{split}
    p&(c,a,d,b|x,y) =\\
    &p(c)\; p(a|x,c)\; p(d|c,x,a)\; p(b|c,x,a,d,y).
\end{split}
\end{equation}

\smallskip

\textbf{Consequence of \textbf{Axiological Time Symmetry}:} From Bayes' rule we can write,
\begin{equation}
\begin{split}
p(a,d|c,x)&=p(a|x,c)p(d|c,x,a)\\
&=p(d|x,c)p(a|c,x,d)
\end{split}
\end{equation}

Inserting this into Eq.~\eqref{nrc1} for the forward-time scenario where Charlie-Alice measures before Debbie-Bob, we get,
\begin{equation}
\label{eq:fwd}
\begin{split}
    p&_{\rightarrow}(c,a,d,b|x,y)=\\
    &p(c)p(d|x,c)p(a|c,x,d)p(b|c,x,a,d,y).
\end{split}
\end{equation}

Similarly, for the time-reversed scenario, we have,
\begin{equation}
\label{eq:rev}
\begin{split}
    p&_{\leftarrow}(d,b,c,a|y,x)=\\
    &p(d)p(c|y,d)p(b|d,y,c)p(a|d,y,b,c,x).
\end{split}
\end{equation}

This yields two key conditional independences, proved explicitly as Lemma~\ref{lem:deb} and Lemma~\ref{lem:bob}.

\begin{Lemma}
\label{lem:deb}
If \textbf{Absoluteness of Observed Events}, \textbf{Axiological Time Symmetry} and \textbf{No Retrocausality} hold, then Charlie and Debbie's joint outcome statistics are independent of Alice and Bob's choice of measurements $x$ and $y$, i.e.,
\[
p(c,d|x,y) = p(c,d).
\]
\end{Lemma}

\begin{proof}

From the assumption of \textbf{Axiological Time Symmetry} we have,

\begin{align}
\label{eq:timeS}
     &p_{\leftarrow}(d,b,c,a|y,x)=p_{\rightarrow}(c,a,d,b|x,y)
\end{align}

Summing over $a$, and $b$, Eq.~\eqref{eq:fwd} and \eqref{eq:rev} can be written as,
\begin{equation}
\label{eq:fwd2}
\begin{split}
    p_{\rightarrow}(c,d|x,y)&=p(c) p(d|x,c)=p(c,d|x)
\end{split}
\end{equation}
and
\begin{equation}\label{eq:rev2}
         p_{\leftarrow}(d,c|x,y)=p(d) p(c|y,d)=p(c,d|y)
\end{equation}
respectively.

Finally, combining Eqs.~\eqref{eq:fwd2} and \ref{eq:rev2} with Eq.~\eqref{eq:timeS} gives,
\begin{align}
p_{\rightarrow}(c,d|x,y)=
    p(c,d|y)=p(c,d|x)=p(c,d)
\end{align}

\end{proof}

\begin{Lemma}[(Alice(Bob) cannot depend on Bob's(Alice's) choice of measurement beyond $c,d$]
\label{lem:bob}
If \textbf{Absoluteness of Observed Events}, \textbf{Axiological Time Symmetry} and \textbf{No Retrocausality} hold, then Alice's(Bob's) outcome $a$($b$) is effectively screened off by the Pseudo Events in its causal past; that is, it cannot depend on Bob's (Alice's) choice of measurement, and hence,
\begin{align}
p(a|c,d,x,y ) &= p(a|c,x,d) \\
p(b | y, c, x, d) &= p(b | y, c, d).
\end{align}
\end{Lemma}

\begin{proof}

By using Lemma \ref{lem:deb} we can now write the joint probability as,
\begin{equation}
    p(c,a,d,b|x,y)= p(c,d) p(a|c,x,d) p (b|c,x,a,d,y),
\end{equation}
and then summing over $b$ gives,
\begin{equation}
    p(c,a,d|x,y)= p(c,d) p(a|c,x,d).
\end{equation}
Now dividing both sides by $p(c,d|x,y)$ (which by Lemma \ref{lem:deb} is equal to $p(c,d)$) we get,
\begin{align}
    p(a|c,d,x,y )&= \frac{p(c,a,d|x,y)}{p(c,d|x,y)} =p(a|c,x,d)
\end{align}
However, due to \textbf{No Retrocausality} we can always write $p(a|c,x,d)= p(a|c,x)$.

Similarly, for the time-reversed scenario, using the \textbf{Axiological Time Symmetry} assumption (for notational simplicity $'\rightarrow'$ in the subscript is omitted), we can write,

\begin{equation}
    p(c,a,d,b|x,y)= p(c,d) p(b|d,y,c) p (a|d,y,b,c,x),
\end{equation}

Then summing over $a$ and dividing both side by $p(c,d|x,y)$ finally gives,

\begin{align}
    p(b|c,d,x,y)= p(b|c,d,y)
\end{align}

\end{proof}

\smallskip

We have now established two important results in the form of Lemmas~\ref{lem:deb} and~\ref{lem:bob}, which will be used to prove the main result. At this stage, we insert the \textbf{Screening via Pseudo Events} assumption (i.e., $p(b|c,x,a,d,y)=p(b|c,x,d,y) \ \ \forall b,c,x,d,y$) into the causal factorisation (Eq.~\eqref{nrc1}), then inserting Lemmas~\ref{lem:deb} and~\ref{lem:bob}, and \textbf{No Retrocausality} assumption to the modified expression yields
\begin{align}
p(c,a,d,b|x,y) &= p(c,d)\;p(a|x,c) \;p(b|y,c,d)\,.
\end{align}
Marginalising over $c$ and $d$, and writing the probabilities $p(c)p(d|c)=p(c,d)$, gives the observed joint probability distribution as
\begin{equation}
p(a,b|x,y) = \sum_{c,d} p(c,d)\;p(a|x,c)\;p(b|y,c,d).
\label{eq:joint}
\end{equation}

\smallskip

For each fixed pair $(c,d)$, we can define conditional expectation values
\begin{equation}
\begin{split}
    A_x^{c} &:= \sum_{a=\pm1} a\,p(a|x,c) \in [-1,1],\\ 
B_y^{cd} &:= \sum_{b=\pm1} b\,p(b|y,c,d) \in [-1,1],
\end{split}
\end{equation}
and so correlators
\begin{equation}
\begin{split}
\langle& A_{x} B_{y} \rangle = \sum_{a,b=\pm1} ab\,p(a,b|x,y)\\
&= \sum_{c,d} p(c,d)\,\Big[ \sum_a a\,p(a|x,c) \Big]\Big[\sum_b b\,p(b|y,c,d)\Big]\\
&= \sum_{c,d} p(c,d)\,A_x^{c}\,B_y^{cd}.
\end{split}
\end{equation}

\smallskip

For each $(c,d)$, we define the functional,
\begin{equation}
\begin{split}
 S^{cd} &:= A_0^{c} B_0^{cd} + A_0^{c} B_1^{cd} + A_1^{c} B_0^{cd} - A_1^{c} B_1^{cd} \\
 &= A_0^{c}(B_0^{cd}+B_1^{cd}) + A_1^{c}(B_0^{cd}-B_1^{cd}).
\end{split}
\end{equation}

Since $B_y^{cd}\in[-1,1]$, and 
\begin{equation}
\forall u,v\in[-1,1],\,|u+v| + |u-v|\le 2
\end{equation}
we can maximise over $|A_0^{c}|\le 1$, $|A_1^{c}|\le 1$ to get
\begin{equation}
|S^{cd}| \le |B_0^{cd}+B_1^{cd}| + |B_0^{cd}-B_1^{cd}| \le 2.
\end{equation}

Therefore, the observed value of the function $S$ is
\begin{equation}
    \begin{split}
       S&= \sum_{c,d} p(c,d)\; S^{cd}\\
  &= \sum_{c,d} p(c,d)\,\Big[ A_0^{c} B_0^{cd} + A_0^{c} B_1^{cd} + A_1^{c} B_0^{cd} - A_1^{c} B_1^{cd} \Big]\\
  &= \langle A_{0} B_{0} \rangle +\langle A_{0} B_{1} \rangle +\langle A_{1} B_{0} \rangle -\langle A_{1} B_{1} \rangle , 
    \end{split}
\end{equation}
so by Jensen's inequality \cite{Nielsen_Chuang_2010},
\begin{equation}
\label{eq:inequality}
  |S| \;\le\; \sum_{c,d} p(c,d)\,|S^{cd}| \;\le\; \sum_{c,d}p(c,d) \times \,2 = 2.  
\end{equation}

\textbf{Quantum violation:}  
To achieve the maximal violation of the inequality in Eq.~\ref{eq:inequality}, Charlie and Alice must perform measurements of complementary observables, as should Debbie and Bob. Under such choices, the quantum value reaches  
\begin{equation}
  S_{\mathrm{QM}} = 2\sqrt{2} > 2.  
\end{equation}

{For example, one optimal choice of measurement settings is:
\begin{equation}
\begin{aligned}
&\text{Charlie: } A_{0} =\sigma_{z}, \quad &&\text{Alice: } A_{1}= \sigma_{x}, \\
&\text{Debbie: } B_{0}= \frac{\sigma_{z} + \sigma_{x}}{\sqrt{2}}, \quad &&\text{Bob: } B_{1}=\frac{\sigma_{z} - \sigma_{x}}{\sqrt{2}}.
\end{aligned}
\end{equation}}

{On the other hand, the preparation stage corresponds to the eigenstates of Alice's measurements. Explicitly, for input \(x=0\), the preparation ensemble consists of the eigenstates of \(A_0 = \sigma_z\), while for \(x=1\), it consists of the eigenstates of \(A_1 = \sigma_x\), i.e.
\begin{equation}
\rho_{a|x} = \frac{1}{2}(I + a A_x), \ \ \ x=0,1
\end{equation}
Averaging over outcomes yields
\begin{equation}
\rho_A = \sum_a \rho_{a|x} = \frac{I}{2} \quad \forall x,
\end{equation}
showing that the preparation is no-signalling in time. It is proved in \cite{leifer2017} that such no-signalling quantum preparations implies Operational Time-Symmetric quantum behaviour. This ensures operational Time-Symmetry of the construction. Therefore, no theory satisfying \textbf{Absoluteness of Observed Events} and \textbf{Causal Agency} can reproduce all quantum Causal Time-Symmetric Friendliness correlations.}

\section{Causal Time-Symmetric Friendliness inequality from relaxed assumptions: Proof of Theorem \ref{thm:2}}\label{App:ProofTheorem2}

 In order to prove that the factorisation of operational probability $p(a,b|x,y)$ as given in Eq.~\eqref{eq:opem_factorisation}, we recall the operational ingredients posited by \textbf{Operational Pseudo Event Mediation} - specifically,
 \begin{equation}
 \label{eq:SPE}
  p(b | a,c,d,x,y)=p(b | c,d,y)\qquad\forall a,b,c,d,x,y.
\end{equation}

Crucially, we do not assume the prior existence of a Kolmogorovian joint probability measure $p(a,b,c,d| x,y)$ defined by all the four events. Instead construct the probability from the operational marginals that are empirically defined. To this end we define such a operationally constructed distribution for operational joint probability-like object as
\begin{equation}
\label{eq:q_def}
  q(a,b,c,d | x,y) \;:=\; p(c,d | x,y)\; p(a| x,c)\; p(b | y,c,d).
\end{equation}

The right-hand side of Eq.~\eqref{eq:q_def} uses only the operational objects that \textbf{Operational Pseudo Event Mediation} assumes exist.  Therefore, $q$ should be regarded as a \emph{constructed} function of these empirical distributions, rather than as a postulate implying the prior existence of a full joint distribution. However, we expect that for every fixed $(x,y)$, $q(\cdot,\cdot,\cdot,\cdot | x,y)$ is a non-negative function which sums up to unity as each factor constitutes a conditional probability:
    \begin{eqnarray}
      &&\sum_{a,b,c,d} q(a,b,c,d | x,y) \nonumber\\
      &&=\sum_{c,d} p(c,d | x,y)\Big(\sum_a p(a| x,c)\Big)\Big(\sum_b p(b | y,c,d)\Big) \nonumber\\
      &&=\sum_{c,d} p(c,d | x,y)=1.
    \end{eqnarray}

Hence $q(\cdot,\cdot,\cdot,\cdot| x,y)$ is a valid probability distribution. Defining $q$ in
Eq.~\eqref{eq:q_def} is \emph{not} the same as assuming the existence of the previously given Kolmogorovian joint probability. Rather, $q$ is built from marginals and conditionals which are assumed to be empirically defined entities; this is a strictly weaker commitment than postulating a global joint probability a priori.

Now the statistics corresponding to the Truly-Observed Events $a$ and $b$ can be recovered by marginalising $q$ as,
    \begin{eqnarray}
    \label{eq:reconstruct}
      &&\sum_{c,d} q(a,b,c,d | x,y) \nonumber\\
      &&= \sum_{c,d} p(c,d | x,y)\,p(a | x,c)\,p(b | y,c,d) \nonumber\\
      &&:= \widetilde p(a,b | x,y).
    \end{eqnarray}

By construction $\widetilde p(a,b|x,y)$ is the operational prediction obtained by combining the assumed marginals in \textbf{Operational Pseudo Event Mediation}. To prove that $\widetilde p(a,b | x,y)$ is same as $p(a,b|x,y)$ given in Eq.~\eqref{eq:opem_factorisation} we have to prove the independence of $c,d$ from $x,y$. However, we cannot establish this result in the same manner as Lemma~\ref{lem:deb}, since the proof of Lemma~\ref{lem:deb} relies on the assumption of \textbf{Absoluteness of Observed Events}, which we do not use here. Using instead the operational mediation assumption of \textbf{Operational Pseudo Event Mediation}, we write the constructed distribution as in Eq.~\eqref{eq:q_def}. 

Further, \textbf{Axiological Time Symmetry} implies that the forward and time-reversed operational constructions coincide under the exchange $(c,a,x)\!\leftrightarrow\!(d,b,y)$, meaning we can write,
\begin{equation}
\label{eq:q_alt2}
  q(a,b,c,d | x,y)
  = p(c,d | y,x)\,p(b | y,d)\,p(a | x,c,d),
\end{equation}
where the quantities on the right-hand side are again operational marginals defined for the time-reversed experiment. Note that, every object that appears in the algebra above (namely, the factors in $q$ and the equalities given by \textbf{Axiological Time Symmetry}) is either an operational marginal (e.g. $(p(c,d | x,y)$) or an operational conditional (e.g. $p(a|x,c)$, $p(b|y,c,d)$), all assumed to exist by \textbf{Operational Pseudo Event Mediation}. We never invoked any property that would require these to be derived from the same pre-existing global Kolmogorovian probability distribution; rather we constructed $q$ from the operational probabilities.
Equating Eqs.~\eqref{eq:q_def} and~\eqref{eq:q_alt2} gives
\begin{equation}
\label{eq:q_equate}
\begin{split}
    p(c,d | x,y)\,&p(a | x,c)\,p(b | y,c,d) \\
  &= p(c,d | y,x)\,p(b | y,d)\,p(a | x,c,d).
\end{split}
\end{equation}

Summing over $a$ and $b$ we get from Eq.~\eqref{eq:q_equate} $p(c,d|x,y)=p(c,d|y,x)$. Finally, applying \textbf{No Retrocausality} in the forward and time reversed scenario, we get $p(c,d|x,y)=p(c,d|x)$ and $p(c,d|y,x)=p(c,d|y)$. This simply implies that they both are equal to $p(c,d)$, completing the proof.

\end{document}